\documentclass[lettersize,journal]{IEEEtran}
\usepackage{amsfonts}
\usepackage{amsmath}
\usepackage{algorithmic}
\usepackage{algorithm}
\usepackage{array}
\usepackage[caption=false,font=normalsize,labelfont=sf,textfont=sf]{subfig}
\usepackage{textcomp}
\usepackage{stfloats}
\usepackage{url}
\usepackage{verbatim}
\usepackage{graphicx}
\usepackage{cite}
\usepackage{pstool}
\usepackage{psfrag}
\usepackage{amsthm}
\usepackage{amssymb}

\newtheorem{theorem}{Theorem}

\newtheorem{lemma}{Lemma}
\newtheorem{remark}{Remark}
\makeatletter
\def\endthebibliography{%
  \def\@noitemerr{\@latex@warning{Empty `thebibliography' environment}}%
  \endlist
}
\makeatother
\begin{document}

\title{UAV-Mounted IRS (UMI) in the Presence of Hovering Fluctuations: 3D Pattern Characterization and Performance Analysis}

\author{Mohammad Javad Zakavi, Mahtab Mirmohseni,~\IEEEmembership{Senior Member,~IEEE,} Farid Ashtiani,~\IEEEmembership{Senior Member,~IEEE,} Masoumeh Nasiri-Kenari,~\IEEEmembership{Senior Member,~IEEE}

\thanks{Mohammad Javad Zakavi, Farid Ashtiani, and Masoumeh Nasiri-Kenari are with the Department of Electrical Engineering, Sharif University of Technology, Tehran, Iran (e-mail: mohammadjavad.zakavi@ee.sharif.edu, ashtianimt@sharif.edu, mnasiri@sharif.edu).}

\thanks{Mahtab Mirmohseni is with the Institute for Communication Systems, Department of Electrical Engineering, University of Surrey, Guildford, U.K (e-mail: m.mirmohseni@surrey.ac.uk).}}

%

\maketitle

\begin{abstract}
This paper investigates unmanned aerial vehicle (UAV)-mounted intelligent reflecting surfaces (IRS) to leverage the benefits of this technology for future communication networks, such as 6G. Key advantages include enhanced spectral and energy efficiency, expanded network coverage, and flexible deployment. One of the main challenges in employing UAV-mounted IRS (UMI) technology is the random fluctuations of hovering UAVs. Focusing on this challenge, this paper explores the capabilities of UMI with passive/active elements affected by UAV fluctuations in both horizontal and vertical angles, considering the three-dimensional (3D) radiation pattern of the IRS. The relationship between UAV fluctuations and IRS pattern is investigated by taking into account the random angular vibrations of UAVs. A tractable and closed-form distribution function for the IRS pattern is derived, using linear approximation and by dividing it into several sectors. In addition, closed-form expressions for outage probability (OP) are obtained using the central limit theorem (CLT) and the Gamma approximation. The theoretical expressions are validated through Monte Carlo simulations. Our findings indicate that the random fluctuations of hovering UAVs have a notable impact on the performance of UMI systems. To avoid link interruptions due to UAV instability, IRS should utilize fewer elements, even though this leads to a decrease in directivity. As a result, unlike terrestrial IRS, incorporating more elements into aerial IRS systems does not necessarily improve performance due to the fluctuations in UAVs. Numerical results show that the OP can be minimized by selecting the optimal number of IRS elements and using active elements.
\end{abstract}


\section{Introduction}
\IEEEPARstart{T}{he} intelligent reflecting surface (IRS) is a promising technology suggested for future communication networks. In recent years, this technology has gained significant attention due to its capability to manage the wireless environment between transceivers \cite{basharat2021reconfigurable,pan2022overview}. Specifically, an IRS consists of numerous passive/active elements that act as smart reflectors, reflecting the incoming signals in the desired direction by adjusting their phase shifts with a programmable controller. As a result, the signals reflected by individual IRS elements can be combined with multipath signals, thereby amplifying the received signal strength and reducing interference. This process contributes to a notable improvement in spectral efficiency (SE) \cite{zhu2022intelligent}. Unlike conventional active relays like decode-and-forward (DF) and amplify-and-forward (AF) relays, the passive IRS is composed of passive components and does not require RF chains. Furthermore, IRS reflects incoming signals with minimal power usage, leading to low energy consumption \cite{bjornson2020intelligent}. Because of their straightforward installation process and flexible shape, IRS can be installed on building facades or directly onto drones to help terrestrial user equipment (UEs) with poor connectivity \cite{shakhatreh2024mobile,pogaku2022uav}. The features mentioned above make IRS not only a promising solution for green networks but also a valuable complement to other existing technologies. However, the performance of passive IRS-assisted systems may be restricted by significant product-distance path loss. One potential solution to this challenge involves increasing the number of passive reflecting elements to leverage the square-order beamforming gain. Alternatively, deploying a passive IRS in proximity to the transmitter and/or receiver can effectively mitigate path loss \cite{wu2021intelligent}. On the other hand, a novel variant of IRS, known as active IRS, has recently been introduced in \cite{kang2024active,zhang2023active,zhi2022active,long2021active}. This approach addresses the drawbacks of passive IRS by employing low-cost components to amplify the reflected signals. Typically, an active IRS comprises several active reflecting elements that integrate negative resistance components like tunnel diodes and negative impedance converters. This allows the IRS to reflect incoming signals while boosting its power \cite{lonvcar2020challenges}. In contrast to AF relays, which depend on power-hungry RF chains, active IRS reflects signals directly through low-power reflection-type amplifiers.
\IEEEpubidadjcol

In recent years, unmanned aerial vehicles (UAVs) have become increasingly popular for applications such as surveillance, tracking, remote sensing, and disaster communication, owing to their cost-efficiency, autonomous functionality, and straightforward deployment. In cellular networks, UAVs are particularly valuable as they can act as aerial relays or flying base stations, extending coverage and improving connectivity in areas with obstructions or high congestion \cite{geraci2022will}. Nevertheless, their adoption in high-throughput communication systems is hindered by limitations related to size, weight, and power \cite{zeng2019accessing}. The combination of IRS technology with UAVs has emerged as a viable approach to enhance both SE and energy efficiency (EE), generating considerable research attention in this domain \cite{zhang2022irs,zhao2023energy,ji2022trajectory,wei2020sum,cai2022resource,su2022spectrum,qian2023joint,zhao2024energy,cheng2024aerial,ge2023active,adam2023intelligent}. In the next subsection, we provide a comprehensive review of related works, focusing on the integration of UAVs and IRS technology.

\subsection{Related Works}
A downlink communication system employing an IRS-assisted UAV is discussed in \cite{zhang2022irs}, where the UAV dynamically establishes a cascaded link through the IRS to improve signal quality for multiple users. The study focuses on maximizing the sum rate for all users by optimizing resource allocation, IRS phase shifts, and UAV trajectory using the block coordinate descent method. Key constraints, including transmit power, flight speed, service area, and IRS reflection capabilities, are taken into account. Additionally, \cite{zhao2023energy} explores the integration of IRS with UAVs to achieve energy-efficient communication. The optimization of UAV trajectory, power allocation, and IRS phase shifts is based on statistical channel state information (CSI) to maximize EE. In \cite{ji2022trajectory}, the authors investigate a UAV communication system enhanced by an IRS, concentrating on data transmission from a ground node to the UAV in a jamming environment. To maximize the average communication rate, they propose a joint optimization framework for the ground node's transmit power, IRS passive beamforming, and UAV trajectory. In \cite{wei2020sum}, the research aims to maximize the sum rate of a multi-user IRS-assisted UAV system utilizing orthogonal frequency division multiple access (OFDMA), while satisfying the quality of service (QoS) requirements for all users. The authors in \cite{cai2022resource} consider the use of non-orthogonal multiple access (NOMA) in an IRS-supported UAV network to efficiently serve multiple ground users. In that work, the optimization of resource allocation, 3D UAV trajectory, and IRS phase control are used to minimize the average total energy consumption. These studies share a common feature: they utilize a separate IRS and UAV setup, where the UAV acts as an aerial base station, and the IRS is positioned independently, such as on buildings or near the UAV/UE.

Beyond separate IRS-UAV systems, recent research has introduced UAV-mounted IRS (UMI) configurations, where the IRS is attached to or carried by the UAV. The key advantage of the UMI is its significantly lower power consumption compared to active relays. This is because the UMI reflects signals without needing power-hungry components, requiring only minimal energy for its phase-shift controller \cite{shafique2020optimization,hashida2020intelligent}. While this power draw is small relative to the UAV's propulsion, it is still a factor; thus, mission time can be extended using energy-efficient control algorithms and renewable energy sources like solar panels \cite{liu2022intelligent}. In \cite{su2022spectrum},  a communication scheme is proposed where an IRS is mounted on a UAV to facilitate connections between a base station (BS) and ground users. This study introduces methods to enhance both SE and EE by jointly optimizing active beamforming, passive beamforming, and UAV trajectory. In \cite{qian2023joint}, a UAV carrying an IRS is proposed as a relay node to support covert communication systems (CCS). To ensure covertness, the study calculates the minimum error detection probability at the eavesdropper and optimizes the UAV trajectory and IRS phase shifts to maximize the average communication rate. In \cite{zhao2024energy}, the authors explore the security of simultaneous wireless information and power transfer (SWIPT) systems supported by IRS and UAVs, taking into account the energy consumption of rotary-wing UAVs during flight. In this setup, an IRS mounted on a UAV improves the quality of legitimate transmissions, while artificial noise generated at the base station (BS) is used to disrupt eavesdropping attempts. Ground devices (GDs) utilize power splitting (PS) technology to concurrently decode information and harvest energy. The study focuses on optimizing BS transmit beamforming, UAV-IRS phase shifts, trajectory/velocity, and GD PS ratios to maximize the overall secrecy rate for all GDs.

Recently, UAV jitter in the presence of IRS is studied in \cite{cheng2024aerial,ge2023active,adam2023intelligent}. The authors in \cite{cheng2024aerial} analyze the use of a multi-aerial IRS in a secure SWIPT system, considering UAV jitter. Angle estimation errors caused by UAV jitter are converted into CSI errors using linear approximation techniques. Subsequently, a joint optimization problem is formulated to maximize the average secrecy rate, incorporating the beamforming vector, IRS phase shift matrices, and UAV trajectories. Random airflow and fuselage vibrations can greatly affect the communication capabilities of UAVs. In \cite{ge2023active}, a novel active IRS is introduced to address this issue, enabling a secure and energy-efficient beamforming design. The proposed framework accounts for the effects of UAV jitter and involves the joint optimization of the active IRS reflection coefficient, UAV-based BS beamforming, and UAV trajectory, while ensuring compliance with worst-case secrecy rate constraints. In \cite{adam2023intelligent}, the authors investigate a UAV-assisted multi-user IRS communication system, aiming to minimize power consumption through the joint optimization of active beamforming, passive beamforming, and UAV trajectories. The study takes into account practical constraints such as UAV jitters and imperfections in hardware components.

To fully leverage the advantages of a UMI in operation, it is crucial to characterize the 3D IRS pattern, determine its distribution, and evaluate the effect of the UAV's random fluctuations. UMI systems experience misalignment between transceivers due to UAV fluctuations, leading to a decrease in the signal-to-noise ratio (SNR) at the receiver, ultimately compromising system reliability. Therefore, optimizing the 3D IRS pattern is essential for maintaining a stable connection during UAV fluctuations. This optimization requires carefully balancing the trade-off between increasing the directivity by increasing the number of IRS elements to counteract path loss on the one hand, and decreasing it to mitigate the impacts of UAV fluctuations on the other hand.

\subsection{Our Contribution}
While prior studies have investigated IRS-assisted UAV communications, the impact of UAV hovering fluctuations remains largely unaddressed \cite{zhang2022irs,zhao2023energy,ji2022trajectory,wei2020sum,cai2022resource,su2022spectrum,qian2023joint,zhao2024energy}. Although recent works have begun to explore UAV jitter in IRS-assisted scenarios \cite{cheng2024aerial,ge2023active,adam2023intelligent}, critical gaps persist. For instance, some prior studies formulate the problem as an optimization of beamforming, IRS phase shifts, and other parameters, yet fail to provide closed-form derivations. Moreover, while the number of IRS elements should be optimized under UAV fluctuations, existing work typically assumes a fixed number and lacks analytical derivations to determine the optimal configuration. To address these limitations, this work provides a comprehensive characterization of the 3D IRS pattern under UAV hovering fluctuations, analyzes its statistical distribution, and derives the optimal number of elements for both passive and active IRS architectures. By incorporating both IRS types and explicitly modeling UAV fluctuation effects, our analysis provides new theoretical insights into UMI system design. Due to the UAV's instability, we assess how its horizontal and vertical angle fluctuations affect system performance. The main contribution of this paper is the development of a 3D IRS pattern, along with the determination of its distribution under UAV fluctuations, and the analysis of the outage probability (OP) of the UMI. By leveraging the derived analytical expressions, we minimize the OP under different levels of UAV fluctuations. This is achieved by selecting an optimal number of IRS elements, which ultimately enhances system reliability. The main contributions of this work are outlined as follows\\
1) We address the issue of instability in real-world UAV platforms, which are susceptible to airflow disturbances and airframe vibrations, leading to horizontal and vertical angular fluctuations. To conduct performance analysis, it is necessary to have the distribution of the IRS 3D pattern. Therefore, we characterize the 3D pattern of the IRS and determine its distribution while considering UAV fluctuations. Note that this distribution is applicable in all systems involving UMI for performance analysis under UAV fluctuations.\\
2) To determine the IRS 3D pattern distribution, which is a product of the array factor and single radiation pattern, we follow a three-step analytical approach. First, we obtain the distribution of the array factor, then we treat the single radiation pattern as deterministic. Finally, we derive the tractable and closed-form probability distribution function (PDF) of the IRS 3D pattern by utilizing the results of the previous two steps. To simplify the analysis, we apply linear approximations to address the error variations in elevation and azimuth angles caused by UAV fluctuations. Furthermore, we introduce a sectoral model for the array factor and single radiation pattern, resulting in a more straightforward approach for determining the distribution of the IRS 3D pattern.\\
3) Using the derived PDF of the IRS 3D pattern, we establish closed-form expressions for the OP of the UMI in two scenarios: Passive-IRS and Active-IRS. For this analysis, we apply the central limit theorem (CLT) and Gamma approximation techniques.\\
4) Finally, the accuracy of the derived analytical expressions is validated through Monte Carlo simulations. The results demonstrate that system performance is heavily influenced by UAV fluctuations, with a notable degradation compared to stable conditions. By using our analytical approach, we are able to determine the optimal number of IRS elements under varying levels of UAV fluctuations, enabling the minimization of OP without requiring extensive simulations.

\subsection{Outline and Notations}
The paper is structured as follows: Section II presents the system model and the 3D IRS pattern considering UAV fluctuations. Section III derives the IRS radiation pattern distribution function. Section IV provides a detailed performance analysis. Section V includes simulation results to validate the analytical ones and examines link performance and IRS pattern optimization. Finally, Section VI concludes the paper.

\emph{Notation}: Throughout the paper, matrices and column vectors are denoted by bold uppercase and bold lowercase letters, respectively. The transpose and Hermitian transpose of a vector or matrix are indicated by $(.)^T$ and $(.)^H$, respectively. $\mathbf{0}_{N}$ denotes the zero vector of size $N$, where every component is $0$. The $N\times N$ identity matrix is expressed by $\mathbf{I}_N$. The Hadamard product and Kronecker product of two vectors are denoted by $\odot$ and $\otimes$, respectively. A real and complex normal distribution with mean $\mu$ and variance $\sigma^2$ is expressed as $\mathcal{N}(\mu,\sigma^2)$ and $\mathcal{CN}(\mu,\sigma^2)$, respectively. The Q-function is denoted by $Q(.)$, while the Dirac delta function is represented by $\delta(.)$. Finally, the expectation operation is expressed as $\mathbb{E}(.)$.

\section{System Model}
A communication system utilizing UMI is illustrated in Fig.~\ref{fig1}. This system involves a BS equipped with $M$ antennas and an IRS with $N$ reflective elements, serving a single-antenna UE. The direct link between BS and UE is assumed to be obstructed.
\begin{figure}[!t]
\centering
\includegraphics{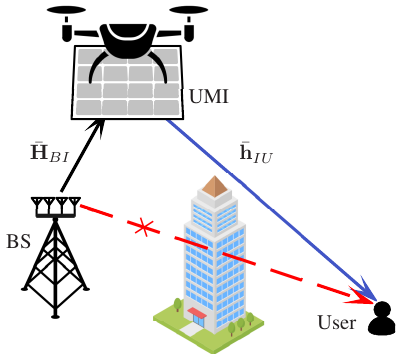}
\caption{The UMI-assisted communication system.}
\label{fig1}
\end{figure}
Without loss of generality, a 3D Cartesian coordinate system is adopted. The positions of the BS, IRS, and UE are defined as $\mathbf{u}_{\footnotesize\text{BS}}=[x_{\footnotesize\text{BS}}, y_{\footnotesize\text{BS}}, h_{\footnotesize\text{BS}}]$, $\mathbf{u}_{\footnotesize\text{IRS}}=[x_{\footnotesize\text{IRS}}, y_{\footnotesize\text{IRS}}, h_{\footnotesize\text{IRS}}]$, and $\mathbf{u}_{\footnotesize\text{UE}}=[x_{\footnotesize\text{UE}}, y_{\footnotesize\text{UE}}, 0]$, respectively. The transmitted signal to IRS and reflected signal from IRS are represented with the subscript $q\in\{t, r\}$ where $t$ and $r$ indicate the transmitted signal and the reflected signal, respectively. As shown in Fig.~\ref{fig2}, $\theta_q$ and $\phi_q$ denote the elevation and azimuth angles from the IRS in the direction of the BS/UE. Assume $\varepsilon_{x}$ and $\varepsilon_{y}$ are the angular variations of UAV along the $x-z$ and $y-z$ axes, respectively. According to the CLT, the fluctuations in the orientations of the UAV have a Gaussian distribution \cite{dabiri20203d,wang2022jittering,dabiri2020analytical}. That is, we have $\varepsilon_{x} \sim \mathcal{N}(\mu_{x},\sigma^2_{x})$, and $\varepsilon_{y} \sim \mathcal{N}(\mu_{y},\sigma^2_{y})$.
\begin{figure}[!t]
\centering
\includegraphics[width=3.4in]{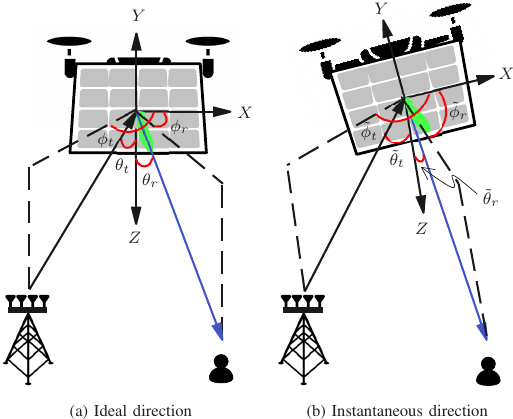}
\caption{A visual depiction of fluctuations in UMI. In this case, the elevation and azimuth without fluctuations are indicated as $\phi_q$ and $\theta_q$, respectively. In contrast, the elevation and azimuth with fluctuations are represented as $\tilde{\phi}_q$ and $\tilde{\phi}_q$, respectively.}
\label{fig2}
\end{figure}
\subsection{Channel Model}
In this paper, we assume a far-field channel model to analyze the impact of UAV fluctuations. Our results demonstrate that these fluctuations inherently constrain the practical optimal size of an IRS. This effective size limitation, combined with the physical mounting constraints on a UAV, justifies the far-field assumption. The systems with extremely large IRS arrays would operate in the near-field, as discussed in \cite{zhi2024performance}; however, our work identifies that UAV fluctuations often preclude such massive scaling in practice. Moreover, we assume the IRS is configured as a uniform square array with $N=N_q \times N_q$ elements. To ensure independence between IRS elements, the spacing between elements in both $x-$ and $y-$directions, denoted as $d_x$ and $d_y$, respectively, is set to half of the wavelength, $\lambda/2$.
The physical dimensions of the IRS are dictated by the operating wavelength and the array configuration. Each meta-atom, with a typical size ($D$) between $\lambda/10$ and $\lambda/5$ \cite{liaskos2018new}, is spaced at $d=\lambda/2$. Consequently, the total aperture size for an $N_q \times N_q$ array is approximately $L \times L$, where $L = (N_q - 1)d + D$. For instance, a $20\times20$ array exhibits a compact form factor of roughly $30\ \text{cm} \times 30\ \text{cm}$ at $10$ GHz and $10\ \text{cm} \times 10\ \text{cm}$ at $30$ GHz. Therefore, a UAV can feasibly carry a moderately-sized IRS comprising a substantial number of elements.

By applying the basic trigonometric formulas, the angles of elevation and azimuth from the IRS towards the direction of the BS/UE without any UAV fluctuations, denoted as $\theta_q$ and $\phi_q$ are obtained as
\begin{subequations}
\begin{align}
&\theta_q = \arctan(\sqrt{\tan^2(\theta_{qx}) + \tan^2(\theta_{qy})}),\\
&\phi_q = \arctan(\frac{\tan(\theta_{qy})}{\tan(\theta_{qx})}),
\end{align}
\label{eq.1}
\end{subequations}
where the directions of the IRS pattern without fluctuations are represented by $\theta_{qx}$ and $\theta_{qy}$ in the $x-z$ and $y-z$ Cartesian coordinates, respectively. On the other hand, the elevation and azimuth angles under UAV fluctuations, i.e., $\tilde{\theta}_q$ and $\tilde{\phi}_q$, are calculated as
\begin{subequations}
\begin{align}
&\tilde{\theta}_q(\varepsilon_x,\varepsilon_y) = \tan^{-1}(\sqrt{\tan^2(\theta_{qx}+\varepsilon_x) + \tan^2(\theta_{qy}+\varepsilon_y)}),\\
&\tilde{\phi}_q(\varepsilon_x,\varepsilon_y) = \tan^{-1}(\frac{\tan(\theta_{qy}+\varepsilon_y)}{\tan(\theta_{qx}+\varepsilon_x)}).
\end{align}
\label{eq.2}
\end{subequations}
where depend on variables $\varepsilon_x$ and $\varepsilon_y$.

Since UMI is usually at altitudes high enough to establish LoS links with ground devices, it experiences low small-scale fading. Consequently, all channels--between BS and IRS, as well as between IRS and UE--are modeled as the Rician distribution \cite{li2023irs}, denoted by $\mathbf{H}_{BI}\in \mathbb{C}^{N\times M}$ and $\mathbf{h}_{IU}\in \mathbb{C}^{N\times 1}$, respectively. The array response of the BS is expressed as $\mathbf{a}_{B}(\theta,\phi)=[1,e^{-j\frac{2\pi}{\lambda}d\sin(\theta)\cos(\phi)},..., e^{-j\frac{2\pi}{\lambda}d(M-1)\sin(\theta)\cos(\phi)}]$ where $d$ is the adjacent antenna distance. The angles of elevation and azimuth from BS towards the direction of the IRS are denoted by $\theta_0$ and $\phi_0$, respectively.
The array response of the IRS is given by
\begin{equation}
\mathbf{a}_I(\theta,\phi) = [\mathbf{a}_{Ix}(\theta,\phi)\otimes\mathbf{a}_{Iy}(\theta,\phi)]^T.
\end{equation}
where
\begin{subequations}
\small
\begin{align}
&\mathbf{a}_{Ix}(\theta,\phi)=[1,e^{-j\frac{2\pi}{\lambda}d_x\sin(\theta)\cos(\phi)},..., e^{-j\frac{2\pi}{\lambda}d_x(N_q-1)\sin(\theta)\cos(\phi)}],\\
&\mathbf{a}_{Iy}(\theta,\phi)=[1,e^{-j\frac{2\pi}{\lambda}dy\sin(\theta)\sin(\phi)},..., e^{-j\frac{2\pi}{\lambda}d_y(N_q-1)\sin(\theta)\sin(\phi)}].
\end{align}
\end{subequations}
Therefore, the channel of the BS-IRS and IRS-UE can be respectively given by
\begin{subequations}
\begin{align}
&\bar{\mathbf{H}}_{BI} = \sqrt{\beta_0E(\zeta_t)}(\mathbf{H}_{BI}\odot (\mathbf{a}_I(\zeta_t)\mathbf{a}_B(\theta_0,\phi_0))),\\
&\bar{\mathbf{h}}_{IU} = \sqrt{\beta_1E(\zeta_r)}(\mathbf{h}_{IU}\odot \mathbf{a}_I(\zeta_r)),
\end{align}
\end{subequations}
where, $\zeta_t = \{\tilde{\theta}_t, \tilde{\phi}_t\}$, $\zeta_r = \{\tilde{\theta}_r,\tilde{\phi}_r\}$ and $E(\zeta)$ is calculated for $\zeta=\{\theta,\phi\}$ as
\begin{equation}
E(\zeta) =
\begin{cases}
  \cos^3(\theta), & \theta \in [0,\pi/2], \phi \in [0,2\pi] \\
  0, & \theta \in (\pi/2,\pi], \phi \in [0,2\pi].
\end{cases}
\label{eq.6c}
\end{equation}
$\beta_0$ and $\beta_1$ are defined as $\beta_k\triangleq c_0{d_k}^{-\alpha_k},\quad k\in\{0,1\}$ where $c_0$ represents the reference path loss at a distance of $1$m. The path loss exponents for BS-IRS and IRS-UE links are denoted by $\alpha_0$ and $\alpha_1$, respectively, whereas $d_0$ and $d_1$ represent the distance from the IRS to the BS and UE, respectively. The Rician factors of $\mathbf{H}_{BI}$ and $\mathbf{h}_{IU}$ are $K_0$ and $K_1$, respectively. We assume the small-scale fading is constant over a channel coherence interval (CCI), which is determined by both multi-path fading and UAV fluctuations. Since UAV fluctuations (e.g., wind or mechanical vibrations) occur on a much larger time scale, their variation within a single CCI is negligible.

\subsection{Received Signal Model}
The reflection matrix of the active IRS is defined as $\mathbf{\Theta}\triangleq\mathbf{A}\mathbf{\Phi}$, where $\mathbf{\Phi}\triangleq \text{diag}(e^{j\phi_1}, ... ,e^{j\phi_N})$ represents the IRS phase shift matrix with $\phi_n$ being the phase shift at the $n$-th element, $n\in \{1, 2, ..., N\}$. Additionally, the active IRS amplification matrix is denoted by $\mathbf{A} \triangleq \text{diag}(a_1, ...,a_N)$ where $a_n$ is the amplification factor of the $n$-th element. Unlike passive IRS, active IRS introduces non-negligible thermal noise at each reflector, denoted as $\mathbf{n}_F \in \mathbb{C}^{N\times1}$, which is typically modeled as $\mathbf{n}_F \sim \mathcal{N}(\mathbf{0}_N, \sigma_f^2\mathbf{I}_N)$, where $\sigma_f^2$ is the amplification noise power \cite{long2021active}. Finally, the received signal at the UE is expressed as
\begin{equation}
y = {\bar{\mathbf{h}}}^{H}_{IU}\mathbf{\Theta} {\bar{\mathbf{H}}}_{BI}\mathbf{w}s + {\bar{\mathbf{h}}}^{H}_{IU}\mathbf{\Theta} \mathbf{n}_F + n,
\end{equation}
where $n\sim\mathcal{N}(0,\sigma_n^2)$ is the additive white Gaussian noise (AWGN) at the UE. The transmit signal from BS is represented by $s$ with transmit power $P_t$. The beamforming vector is denoted by $\mathbf{w} \in \mathbb{C}^{M\times1}$. The power constraint at active IRS is expressed as follows \cite{you2021wireless}
\begin{equation}
P_t\|\mathbf{A\Phi}{\bar{\mathbf{H}}}_{BI}\mathbf{w}\|^2+\|\mathbf{A\Phi} \mathbf{I}_N\|^2 \sigma_f^2 \leq P_F,
\end{equation}
where $P_F$ is the maximum amplification power of the active IRS.

Since the CSI is sensitive to unforeseen feedback delays and IRS effects, the resulting residual CSI error can significantly degrade system performance \cite{hong2020robust,zhang2021robust,chen2022robust}. Given this imperfect CSI at the transmitter, denoted as $\mathbf{h}$, its relationship with the perfect CSI can be modeled as follows \cite{chen2022robust}
\begin{equation}
\hat{\mathbf{h}}_k = \sqrt{1-\zeta}\mathbf{h} + \sqrt{\zeta}\mathbf{h}_e,
\end{equation}
where $\mathbf{h}^H ={\bar{\mathbf{h}}}^{H}_{IU}\mathbf{\Theta} {\bar{\mathbf{H}}}_{BI}$ denotes the perfect cascaded channel, and $\mathbf{h}_e$ is the residual CSI error with zero-mean complex Gaussian entries of variance $\sigma_e^2=E\{|{h_m}|^2\}$, and is independent of $\mathbf{h}$. The CSI error factor $\zeta$, which ranges from $0$ to $1$, indicates the accuracy of the CSI. Accounting for this imperfection, the received signal at the UE is written as
\begin{equation}
\hat{y} = \sqrt{1-\zeta}\mathbf{h}^H\mathbf{w}s + \sqrt{\zeta}\mathbf{h}^H_e\mathbf{w}s + {\bar{\mathbf{h}}}^{H}_{IU}\mathbf{\Theta} \mathbf{n}_F + n.
\end{equation}
To provide a conservative performance bound ensuring reliable system design, we model the CSI error effect by incorporating its variance in the effective noise power $\sigma_e^2$. Therefore, the instantaneous SNR at the UE can be expressed as
\begin{equation}
\gamma = \frac{P_t(1-\zeta)|{\bar{\mathbf{h}}}^{H}_{IU}\mathbf{A\Phi}{\bar{\mathbf{H}}}_{BI}\mathbf{w}|^2}{ \|\bar{\mathbf{h}}^{H}_{IU}\mathbf{A\Phi}\|^2\sigma_f^2 + P_t\zeta \|\mathbf{w}\|^2\sigma_e^2 + \sigma_n^2}.
\label{eq.11}
\end{equation}
As discussed in \cite{yao2023universal} and \cite{wang2020intelligent}, we can assume a strong LoS component and negligible NLoS components for the BS-IRS link. This is due to the strong directivity of the multiple antennas at the BS and the low distance between the BS and the IRS. Based on the discussion in \cite{you2021wireless}, for simplicity, we assume that each reflecting element has the same amplification factor $A$. To achieve maximum SNR, $\mathbf{w}^*$ is determined by the maximum ratio transmission (MRT) principle as \cite{yao2023universal}
\begin{equation}
\mathbf{w}^*=\frac{(\bar{\mathbf{h}}^{H}_{IU}\mathbf{A}^*\mathbf{\Phi}^* \bar{\mathbf{H}}_{BI})^H}{\|\bar{\mathbf{h}}^{H}_{IU}\mathbf{A}^*\mathbf{\Phi}^* \bar{\mathbf{H}}_{BI}\|},
\end{equation}
which satisfies the unit power constraint $\|\mathbf{w}^*\|^2=1$. The optimal active-IRS reflection design is
\begin{equation}
[\mathbf{\Phi}^*]_n=e^{-j(\angle[{\mathbf{H}}_{BI}]_{m,n}+\angle[{\mathbf{h}}_{IU}^H]_n+[\mathbf{a}_I(\theta_t,\phi_t)]_n+[\mathbf{a}_I(\theta_r,\phi_r)]_n)},~\forall n,
\label{eq.70}
\end{equation}
\begin{equation}
(A^*)^2=\frac{P_F}{P_t\beta_0 E(\zeta_t)\sum_{n=1}^{N}|H_{BI,m,n}|^2+N\sigma_f^2}.
\label{eq.71}
\end{equation}
Hence, the instantaneous maximum SNR can be expressed as
\begin{equation}
\gamma_{\text{max}} = \frac{P_tM(1-\zeta)\beta_0\beta_1(A^*)^2\mathbb{G}(\sum_{n=1}^{N}|H_{BI,m,n}||h_{IU,n}|)^2}{\beta_1(A^*)^2E(\zeta_r)\sigma_f^2\sum_{n=1}^{N}|h_{IU,n}|^2 + P_t\zeta \sigma_e^2 + \sigma_n^2},
\label{eq.12c}
\end{equation}where $\mathbb{G}$ describes the array radiation pattern of the IRS towards angles $\tilde{\theta}_t$ and $\tilde{\phi}_t$ from IRS to BS, along with $\tilde{\theta}_r$ and $\tilde{\phi}_r$ from IRS to UE, scaling the received power at the user, defined as \cite{tang2021wireless}
\begin{equation}
\mathbb{G}(N_q,\zeta_t,\zeta_r) = \mathbb{G}_e(\zeta_t,\zeta_r) \times \mathbb{G}_a(N_q,\zeta_t,\zeta_r),
\label{eq.3}
\end{equation}
where $\zeta_t = \{\tilde{\theta}_t, \tilde{\phi}_t\}$, $\zeta_r = \{\tilde{\theta}_r,\tilde{\phi}_r\}$. Here, $\mathbb{G}_a(.)$ is the normalized array factor and $\mathbb{G}_e(.)$ is the single-element radiation pattern. Because of UAV fluctuations, it is deduced from (\ref{eq.3}) that the normalized radiation pattern of IRS depends on two independent RVs $\varepsilon_x$ and $\varepsilon_y$. The single-element radiation pattern is obtained as \cite{tang2021wireless}
\begin{equation}
\mathbb{G}_e(\zeta_t, \zeta_r) = E(\zeta_t) \times E(\zeta_r),
\label{eq.4}
\end{equation}
where $E(\zeta)$ for $\zeta=\{\theta,\phi\}$ is calculated from (\ref{eq.6c}).
Furthermore, the normalized array factor can be expressed as (\ref{eq.6}) shown at the top of next page, where $\delta_1=-\sin({{\theta}}_t)\cos({{\phi}}_t)-\sin({{\theta}}_r)\cos({{\phi}}_r)$ and $\delta_2=-\sin({{\theta}}_t)\sin({{\phi}}_t)-\sin({{\theta}}_r)\sin({{\phi}}_r)$  \cite{tang2021wireless}.
\begin{figure*}[h]
\begin{equation}
\mathbb{G}_a = |{\frac{\sin(\frac{N_q\pi}{\lambda}(\sin({\tilde{\theta}}_t)\cos({\tilde{\phi}}_t)+\sin({\tilde{\theta}}_r)\cos({\tilde{\phi}}_r)+\delta_1)d_x)}{N_q\sin(\frac{\pi}{\lambda}(\sin({\tilde{\theta}}_t)\cos({\tilde{\phi}}_t)+\sin({\tilde{\theta}}_r)\cos({\tilde{\phi}}_r)+\delta_1)d_x)}}\times{\frac{\sin(\frac{N_q\pi}{\lambda}(\sin({\tilde{\theta}}_t)\sin({\tilde{\phi}}_t)+\sin({\tilde{\theta}}_r)\sin({\tilde{\phi}}_r)+\delta_2)d_y)}{N_q\sin(\frac{\pi}{\lambda}(\sin({\tilde{\theta}}_t)\sin({\tilde{\phi}}_t)+\sin({\tilde{\theta}}_r)\sin({\tilde{\phi}}_r)+\delta_2)d_y)}}|^2
\label{eq.6}
\end{equation}
\hrulefill
\end{figure*}

\begin{remark}
When the UAV is stable (i.e., experiences no fluctuations), $\mathbb{G}_a = 1$, whereas during fluctuations (e.g., due to wind or dynamic movement), $\mathbb{G}_a < 1$, leading to a decrease in the power of the reflected signal directed to the UE. Furthermore, as the number of IRS elements, $N$, increases, the IRS pattern becomes narrower. While this enhances directivity under stable conditions, it also makes the system more sensitive to UAV fluctuations. Even minor UAV fluctuations can cause significant beam misalignment, significantly reducing $\mathbb{G}_a$. Thus, a fundamental trade-off exists between employing larger IRS arrays to achieve higher gain and maintaining robustness against UAV fluctuations.
\label{remark1}
\end{remark}

\begin{remark}
The instantaneous SNR for a passive IRS is obtained by $A^*=1$ and $\sigma_f^2=0$. Substituting these values into \eqref{eq.12c} yields
\begin{equation}
  \gamma_{\text{max}} = \frac{P_tM(1-\zeta)\beta_0\beta_1\mathbb{G}(\sum_{n=1}^{N}|H_{BI,m,n}||h_{IU,n}|)^2}{P_t\zeta \sigma_e^2 + \sigma_n^2}.
\label{eq.16c}
\end{equation}
\end{remark}

\section{IRS Radiation Pattern Distribution Function}
The maximum SNR in (\ref{eq.12c}) depends on the IRS radiation pattern, $\mathbb{G}$. To analyze the performance, it is necessary to have the distribution of $\mathbb{G}$. Therefore, in this section, we derive the PDF of the IRS radiation pattern. From (\ref{eq.3}), it follows that $\mathbb{G}=\mathbb{G}_e\times\mathbb{G}_a$. To obtain the distribution of $\mathbb{G}$, we perform three steps. First, we derive the distribution of $\mathbb{G}_a$, and subsequently, we demonstrate that UAV fluctuations induce negligible changes in $\mathbb{G}_e$, thus it can be treated as deterministic. Finally, we obtain the distribution of $\mathbb{G}$ by applying the results from the last two steps.

\subsection*{\textbf{Step1}: Deriving the distribution of $\mathbb{G}_a$}
To simplify $\mathbb{G}_a$ given in (\ref{eq.6}), we define $Z_x$ and $Z_y$ as
\begin{subequations}
\begin{align}
&Z_x\triangleq\sin({\tilde{\theta}}_t)\cos({\tilde{\phi}}_t)+\sin({\tilde{\theta}}_r)\cos({\tilde{\phi}}_r)+\delta_1,\\ &Z_y\triangleq\sin({\tilde{\theta}}_t)\sin({\tilde{\phi}}_t)+\sin({\tilde{\theta}}_r)\sin({\tilde{\phi}}_r)+\delta_2,
\end{align}
\label{eq.13c}
\end{subequations}
where $\delta_1=-\sin({{\theta}}_t)\cos({{\phi}}_t)-\sin({{\theta}}_r)\cos({{\phi}}_r)$ and $\delta_2=-\sin({{\theta}}_t)\sin({{\phi}}_t)-\sin({{\theta}}_r)\sin({{\phi}}_r)$. Therefore, we can rewrite (\ref{eq.6}) as follows:
\begin{equation}
G_{a}(Z_x,Z_y)=\underbrace{|\frac{\sin(\frac{N_q\pi}{2}Z_x)}{N_q\sin(\frac{\pi}{2}Z_x)}|^2}_{g_{ax}}\times\underbrace{|\frac{\sin(\frac{N_q\pi}{2}Z_y)}{N_q\sin(\frac{\pi}{2}Z_y)}|^2}_{g_{ay}},
\label{eq.21}
\end{equation}
where $g_{ax}$ and $g_{ay}$ are the array factors along the $x$-axis and $y$-axis, respectively. To approximate $Z_x$ and $Z_y$ in (\ref{eq.21}) defined in (\ref{eq.13c}), we employ linear approximation. For a function of one variable $f(x)$ whose first-order derivative $f_x$ exists at a point $a$, the  $1$nd-order Taylor polynomial near the point $a$ is
\begin{equation}
f(x)\approx f(a) + f_x(a)(x-a).
\label{eq.13b}
\end{equation}
We define $\varepsilon_{\theta_q}(\varepsilon_x,\varepsilon_y) \triangleq {\tilde{\theta}}_q(\varepsilon_x,\varepsilon_y) - \theta_q$ and $\varepsilon_{\phi_q}(\varepsilon_x,\varepsilon_y) \triangleq {\tilde{\phi}}_q(\varepsilon_x,\varepsilon_y) - \phi_q$, representing the fluctuations of the elevation and azimuth angles, respectively. Consequently, ${\tilde{\theta}}_q$ and ${\tilde{\phi}}_q$ in (\ref{eq.13c}) are given by $\theta_q +\varepsilon_{\theta_q}$ and $\phi_q+\varepsilon_{\phi_q}$, respectively. Utilizing (\ref{eq.13b}), we can approximate the values of $\sin(\tilde{x})$ and $\cos(\tilde{x})$ in (\ref{eq.13c}) with $\tilde{x}=x+\varepsilon$, where $\varepsilon$ is near zero, as follows
\begin{subequations}
\begin{align}
&\sin(\tilde{x}) \approx \sin(x) + \varepsilon \cos(x),\\
&\cos(\tilde{x}) \approx \cos(x) - \varepsilon \sin(x).
\end{align}
\label{eq.16b}
\end{subequations}
Hence, we approximate $Z_x$ and $Z_y$ given in (\ref{eq.13c}) by applying (\ref{eq.16b}) as follows
\begin{subequations}
\begin{align}
Z_x(\varepsilon_{\theta_t},\varepsilon_{\phi_t},\varepsilon_{\theta_r},\varepsilon_{\phi_r}) &\approx \cos\theta_t\cos\phi_t\varepsilon_{\theta_t} - \sin\theta_t\sin\phi_t\varepsilon_{\phi_t} \nonumber \\
&+\cos\theta_r\cos\phi_r\varepsilon_{\theta_r} - \sin\theta_r\sin\phi_r\varepsilon_{\phi_r}, \\
Z_y(\varepsilon_{\theta_t},\varepsilon_{\phi_t},\varepsilon_{\theta_r},\varepsilon_{\phi_r}) &\approx \cos\theta_t\sin\phi_t\varepsilon_{\theta_t} + \sin\theta_t\cos\phi_t\varepsilon_{\phi_t} \nonumber \\
&+\cos\theta_r\sin\phi_r\varepsilon_{\theta_r} + \sin\theta_r\cos\phi_r\varepsilon_{\phi_r}.
\end{align}
\label{eq.23}
\end{subequations}
Note that the linear approximation, which enabled our analytical derivation, is valid under small-angle variation conditions (1–2 degrees), a state maintained by stable UAV flight controllers. Our simulation results under these conditions confirm the model accuracy, showing only $1 \%$ approximation error for $\sigma_k^2=1^\circ, k\in\{x,y\}$, which is negligible and does not affect the key insights or conclusions of the paper. A limitation of this approach is that its accuracy may degrade for highly dynamic flight maneuvers with larger attitude variances.

To derive the PDF of $\mathbb{G}_a$, we must determine the distribution of $Z_x$ and $Z_y$, which depend on $\varepsilon_{\theta_q}$ and $\varepsilon_{\phi_q}$, whose distributions are determined in Lemma \ref{lemma1} and Lemma \ref{lemma2}.
\begin{lemma}
The distributions of $\varepsilon_{\theta_q}(\varepsilon_x,\varepsilon_y)$ follows $\varepsilon_{\theta_{q}}\sim\mathcal{N}(\mu_{\varepsilon_{\theta_q}},\sigma^2_{\varepsilon_{\theta_q}})$ where
\begin{subequations}
\begin{align}
&\mu_{\varepsilon_{\theta_q}} = A_{\theta_{qx}}\mu_x+A_{\theta_{qy}}\mu_y,\\ &\sigma^2_{\varepsilon_{\theta_q}}=A^2_{\theta_{qx}}\sigma^2_x+A^2_{\theta_{qy}}\sigma^2_y,
\end{align}
\end{subequations}

\begin{equation}
\begin{aligned}
A_{\theta_{qk}}=\frac{(1+\tan^2\theta_{qk})\tan\theta_{qk}}{\sqrt{\tan^2\theta_{qx}+\tan^2\theta_{qy}}(1+\tan^2\theta_{qx}+\tan^2\theta_{qy})},&\\
k\in\{x,y\}.&
\end{aligned}
\end{equation}
\label{lemma1}
\end{lemma}

\begin{proof}
Please refer to Appendix \ref{appendix.a}.
\end{proof}

\begin{lemma}
The distributions of $\varepsilon_{\phi_q}(\varepsilon_x,\varepsilon_y)$ follows $\varepsilon_{\phi_{q}}\sim\mathcal{N}(\mu_{\varepsilon_{\phi_q}},\sigma^2_{\varepsilon_{\phi_q}})$ where
\begin{subequations}
\begin{align}
&\mu_{\varepsilon_{\phi_q}} = A_{\phi_{qx}}\mu_x+A_{\phi_{qy}}\mu_y,\\
&\sigma^2_{\varepsilon_{\phi_q}}=A^2_{\phi_{qx}}\sigma^2_x+A^2_{\phi_{qy}}\sigma^2_y,
\end{align}
\end{subequations}

\begin{subequations}
\begin{align}
&A_{\phi_{qx}}=-\frac{(1+\tan^2\theta_{qx})\tan\theta_{qy}}{\tan^2\theta_{qx}+\tan^2\theta_{qy}},\\
&A_{\phi_{qy}}=-\frac{(1+\tan^2\theta_{qy})\tan\theta_{qx}}{\tan^2\theta_{qx}+\tan^2\theta_{qy}}.
\end{align}
\end{subequations}
\label{lemma2}
\end{lemma}

\begin{proof}
We approximate $\varepsilon_{\phi_q}$ using linear approximation as follows
\begin{equation}
\varepsilon_{\phi_q}(\varepsilon_x,\varepsilon_y) \approx A_{\phi_{qx}}\varepsilon_x + A_{\phi_{qy}}\varepsilon_y,
\label{eq.18}
\end{equation}
where $A_{\phi_{qx}}$ and $A_{\phi_{qy}}$ are derived following the same procedure outlined in the proof
of Lemma \ref{lemma1} in Appendix \ref{appendix.a}. The distribution of $\varepsilon_{\phi_q}$ is also obtained analogously.
\end{proof}
Since all RVs in (\ref{eq.23}), i.e., $\varepsilon_{\theta_t}$, $\varepsilon_{\phi_t}$, $\varepsilon_{\theta_r}$ and $\varepsilon_{\phi_r}$, follow Gaussian distributions individually, and any linear combination of them also follows a Gaussian distribution, it can be concluded that they are jointly Gaussian RVs. To illustrate, consider the linear combination of them for arbitrary coefficients $\alpha$, $\beta$, $\xi$, and $\kappa$ using (\ref{eq.18}) and (\ref{eq.14}) (See Appendix \ref{appendix.a}) as follows
\begin{equation}
\begin{aligned}
\alpha\varepsilon_{\theta_t}&+\beta\varepsilon_{\phi_t}+\xi\varepsilon_{\theta_r}+\kappa\varepsilon_{\phi_r}=\alpha A_{\theta_{tx}}\varepsilon_x+\alpha A_{\theta_{ty}}\varepsilon_y \\
&+ \beta A_{\phi_{tx}}\varepsilon_x+\beta A_{\phi_{ty}}\varepsilon_y+\xi A_{\theta_{rx}}\varepsilon_x+\xi A_{\theta_{ry}}\varepsilon_y\\
&+\kappa A_{\phi_{rx}}\varepsilon_x+\kappa A_{\phi_{ry}}\varepsilon_y.
\end{aligned}
\label{eq.23c}
\end{equation}
It is clear from (\ref{eq.23c}) that this linear combination is itself a linear function of the independent Gaussian RVs $\varepsilon_x$ and $\varepsilon_y$. Therefore, any such linear combination remains Gaussian.
Therefore, the distributions of $Z_k$ for $k\in\{x,y\}$ follows the Gaussian distribution, specifically $Z_k\sim \mathcal{N}(\mu_{z_k},\sigma^2_{z_k})$ where $\mu_{z_k}$ is calculated as
\begin{equation}
\mu_{z_k}=\mathbf{a}_k^T {\boldsymbol{\mu}}_{\varepsilon}
\end{equation}
where $\mathbf{a}_x=[\cos\theta_t\cos\phi_t,-\sin\theta_t\sin\phi_t,\cos\theta_r\cos\phi_r,\\-\sin\theta_r\sin\phi_r]^T$, $\mathbf{a}_y=[\cos\theta_t\sin\phi_t,\sin\theta_t\cos\phi_t,\cos\theta_r\sin\phi_r,\\\sin\theta_r\cos\phi_r]^T$, and ${\boldsymbol{\mu}}_{\varepsilon}=[\mu_{\varepsilon_{\theta_t}},\mu_{\varepsilon_{\phi_t}},\mu_{\varepsilon_{\theta_r}},\mu_{\varepsilon_{\phi_r}}]^T$.
The variance of $Z_k$, $\sigma^2_{z_k}$ is obtained as
\begin{equation}
\sigma^2_{z_k}=\mathbf{a}_k^T\mathbf{\Sigma}\mathbf{a}_k
\end{equation}
where
\begin{equation}
\scalebox{0.85}{$\mathbf{\Sigma} =  \begin{bmatrix}
                    {\sigma^2_{\varepsilon_{\theta_t}}} & \text{Cov}(\varepsilon_{\theta_t},\varepsilon_{\phi_t}) & \text{Cov}(\varepsilon_{\theta_t},\varepsilon_{\theta_t}) & \text{Cov}(\varepsilon_{\theta_t},\varepsilon_{\phi_r}) \\
                    \text{Cov}(\varepsilon_{\theta_t},\varepsilon_{\phi_t}) & \sigma^2_{\varepsilon_{\phi_t}} & \text{Cov}(\varepsilon_{\phi_t},\varepsilon_{\theta_r}) & \text{Cov}(\varepsilon_{\phi_t},\varepsilon_{\phi_r}) \\
                    \text{Cov}(\varepsilon_{\theta_t},\varepsilon_{\theta_r}) & \text{Cov}(\varepsilon_{\phi_t},\varepsilon_{\theta_r}) & \sigma^2_{\varepsilon_{\theta_r}} & \text{Cov}(\varepsilon_{\theta_r},\varepsilon_{\phi_r}) \\
                    \text{Cov}(\varepsilon_{\theta_t},\varepsilon_{\phi_r}) & \text{Cov}(\varepsilon_{\phi_t},\varepsilon_{\phi_r}) & \text{Cov}(\varepsilon_{\theta_r},\varepsilon_{\phi_r}) & \sigma^2_{\varepsilon_{\phi_r}}
                  \end{bmatrix}$}
\label{eq.24}
\end{equation}
The $\text{Cov}(\varepsilon_m,\varepsilon_n)$ in (\ref{eq.24}) is calculated as
\begin{equation}
\text{Cov}(\varepsilon_m,\varepsilon_n)=A_{mx}A_{nx}\sigma_x^2+A_{my}A_{ny}\sigma^2_y,
\end{equation}
where
$(m,n) \in \{(\theta_t,\phi_t),(\theta_t,\theta_r),(\theta_t,\phi_r),(\phi_t,\theta_r),(\phi_t,\phi_r),\\(\theta_r,\phi_r)\}$.
For simplicity, we approximate the numerator and denominator of $g_{ax}$ in (\ref{eq.21}) when $Z_x$ is near zero as $\frac{1}{2}(1-\cos(N_q\pi Z_x))$ and $\frac{1}{4} N^2_q\pi^2Z^2_x$, respectively. Hence, $g_{ax}$ in (\ref{eq.21}) can be approximate as follows
\begin{equation}
g_{ax}\simeq\frac{2(1-\cos(N_q\pi Z_x))}{N_q^2\pi^2Z_x^2}.
\label{eq.27}
\end{equation}
Following a similar approach, the approximation of $g_{ay}$ can be derived in the same manner as (\ref{eq.27}) by replacing the subscript $x$ with $y$.
Deriving a closed-form PDF for $\mathbb{G}_a$ is mathematically intractable. We cannot easily integrate over this complex function to find a statistical distribution. To solve this, we propose a sectoral model, which is a form of simplification. Therefore, we split $g_{ax}$ in (\ref{eq.27}) into several sectors and propose a simpler model given by
\begin{equation}
\begin{aligned}
&g_{ax}\simeq\Pi(\frac{DN_q|Z_x|}{2})+\sum_{i=1}^{lD-1}{\frac{D^2(1-\cos(\frac{2\pi i}{D}))}{2\pi^2i^2}}\\
&\qquad\qquad\times[\Pi(\frac{DN_q|Z_x|}{2(i+1)})-\Pi(\frac{DN_q|Z_x|}{2i})],
\end{aligned}
\label{eq.28}
\end{equation}
where $\Pi(x)=\begin{cases} 1, & \mbox{if } x\leq 1 \\ 0, & \mbox{if } x> 1 \end{cases}$. $D$ represents the number of sectors, and $l \in \{1, 2\}$ whereby $l = 1$ corresponds to considering only the main lobe of the pattern, while $l = 2$ is used for higher precision by including both the main lobe and the first side lobe. The simplified model for $g_{ay}$ can be obtained similarly to (\ref{eq.28}) by replacing the subscript $x$ with $y$. Fig.~\ref{fig3} plots the exact model and sectoral model of the IRS array factor along the $x$-axis, $g_{ax}$ versus $\varepsilon_x$ for $D = 5$ and $15$, with $N=64$. We can increase the accuracy of the sectoral model by increasing $D$; however, this improvement comes at the cost of increased computational complexity, as noted in Remark \ref{remark6}. As shown in Fig.~\ref{fig3}, when $l=1$, the proposed model only takes into account the main lobe, whereas, for $l=2$, the first side lobe is also included.
\begin{figure}[!t]
\centering
\includegraphics[width=3.5in]{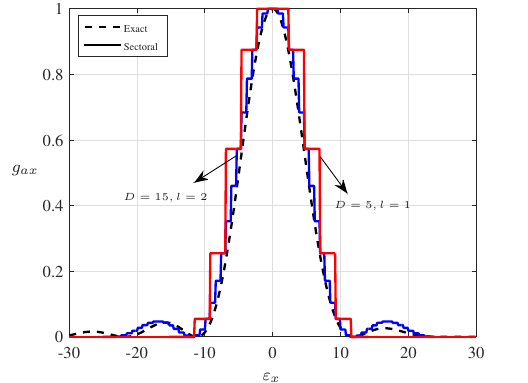}
\caption{The exact model of the IRS array factor along $x$-axis, $g_{ax}$, mentioned in (\ref{eq.21}) and its sectoral model, proposed in (\ref{eq.28}) versus $\varepsilon_x$ for $N=64$.}
\label{fig3}
\end{figure}
\begin{lemma}
The PDF of array factor $G_a$, as defined in (\ref{eq.6}) is derived as
\begin{equation}
f_{\mathbb{G}_a}(\mathbb{G}_a)=\sum_{i=1}^{J}p_a(i)\delta(\mathbb{G}_a-q_a(i))
\label{eq.27b}
\end{equation}
where $J=l^2D^2$, $\mathbf{q}_a = \mathbf{q}_{ax}\otimes \mathbf{q}_{ay}$ and $\mathbf{p}_a = \mathbf{p}_{ax}\otimes \mathbf{p}_{ay}$. Here, $\mathbf{q}_{ak}\in \mathbb{R}^{lD\times1}$ and $\mathbf{p}_{ak}\in \mathbb{R}^{lD\times1}$ where
\begin{equation}
q_{ak}(i)=\frac{D^2(1-\cos(\frac{2\pi i}{D}))}{2\pi^2i^2},
\end{equation}
\begin{equation}
\begin{aligned}
p_{ak}(i)=Q(\frac{2i-DN_q\mu_{Z_k}}{DN_q\sigma_{Z_k}})-Q(\frac{2(i+1)+DN_q\mu_{Z_k}}{DN_q\sigma_{Z_k}})&\\
\qquad+Q(\frac{2i+DN_q\mu_{Z_k}}{DN_q\sigma_{Z_k}})-Q(\frac{2(i+1)-DN_q\mu_{Z_k}}{DN_q\sigma_{Z_k}})&,\\
k\in\{x,y\}&.
\end{aligned}
\end{equation}
\end{lemma}
\begin{proof}
Please refer to Appendix \ref{appendix.b}.
\end{proof}

\subsection*{\textbf{Step2}: Treating $\mathbb{G}_e$ as a Deterministic Component}
Given that $\tilde{\theta}_q=\theta_q+\varepsilon_{\theta_q}$, we can rewrite $\mathbb{G}_e$ mentioned in (\ref{eq.4}) as follows
\begin{equation}
\mathbb{G}_e=\underbrace{\cos^3(\theta_t+\varepsilon_{\theta_t})}_{g_{et}}\times\underbrace{\cos^3(\theta_r+\varepsilon_{\theta_r})}_{g_{er}},
\label{eq.33}
\end{equation}
where $g_{et}$ and $g_{er}$ are the single element radiation patterns in direction of angles $\theta_t$ and $\theta_r$. For a function of one variable $f(x)$ whose first- and second-order derivatives $f^{(1)}(x)$ and $f^{(2)}(x)$, respectively, exist at the point $a$, the  $2$nd-order Taylor polynomial near the point $a$ is
\begin{equation}
f(x)\approx f(a) + f^{(1)}(a)(x-a) + f^{(2)}(a) (x-a)^2.
\label{eq.68b}
\end{equation}
Using (\ref{eq.68b}), for small $\varepsilon_{\theta_q}$, the second-order Taylor expansion around $\theta_q$ is
\begin{equation}
\begin{aligned}
g_{eq} \approx \cos^3(\theta_q)&-3\cos^2(\theta_q)\sin(\theta_q)\varepsilon_{\theta_q}\\ &+\frac{3}{2}\cos(\theta_q)(2-3\cos^2(\theta_q))\varepsilon^2_{\theta_q}
\end{aligned}
\end{equation}
We approximate $\mathbb{G}_e \approx \cos^3(\theta_t)\cos^3(\theta_r)$. The expected error of this approximation is
\begin{equation}
\begin{aligned}
E\{\text{error}\} = &-3\cos^2(\theta_q)\sin(\theta_q)E\{\varepsilon_{\theta_q}\}\\ &+\frac{3}{2}\cos(\theta_q)(2-3\cos^2(\theta_q))E\{\varepsilon^2_{\theta_q}\}
\end{aligned}
\end{equation}
Given that $E\{\varepsilon_{\theta_q}\}=0$ and $E\{\varepsilon^2_{\theta_q}\}$ is very small, the expected error is negligible.
Moreover, $\mathbb{G}_e$ is independent of the number of IRS elements $N$, in contrast to the normalized array factor $\mathbb{G}_a$, which is highly sensitive to $N$. For the subsequent optimization over $N$, we therefore treat $g_{eq}$ as the constant $g_{eq} \approx \cos^3(\theta_q)$. Specifically, the $\mathbb{G}_e$ is treated as
\begin{equation}
\mathbb{G}_e=\cos^3(\theta_t)\times\cos^3(\theta_r).
\label{eq.38c}
\end{equation}
This treatment is consequently applied to $E(\zeta_q)$ for $q\in\{t,r\}$, i.e., $E(\zeta_q)=\cos^3(\theta_q)$.

\subsection*{\textbf{Step3}: Deriving the distribution of $\mathbb{G}$}
Utilizing distribution of $\mathbb{G}_a$ in (\ref{eq.27b}) and treated expression of $\mathbb{G}_e$ in (\ref{eq.38c}), the distribution of $\mathbb{G}$ can be obtained.
\begin{theorem}
The PDF of array radiation gain $\mathbb{G}$ is derived as:
\begin{equation}
f_{\mathbb{G}}(\mathbb{G})=\sum_{i=1}^{J}p(i)\delta(\mathbb{G}-q(i))
\label{eq.38}
\end{equation}
where $J = l^2D^2$, $\mathbf{q} = q_e\mathbf{q}_{a}$, $q_e =\cos^3(\theta_t)\times\cos^3(\theta_r)$, and $\mathbf{p} = \mathbf{p}_{a}$.
\label{theorem1}
\end{theorem}
\begin{proof}
As $\mathbb{G}$ is equal to the product of $\mathbb{G}_e$ and $\mathbb{G}_a$, the PDF of $\mathbb{G}$ is calculated as $f_\mathbb{G}(\mathbb{G}) = \frac{1}{\mathbb{G}_e}f_{\mathbb{G}_a}(\mathbb{G}/\mathbb{G}_e)$. By substituting (\ref{eq.27b}) and (\ref{eq.38c}), the PDF of $\mathbb{G}_e$ is obtained.
\end{proof}

\section{Performance Analysis}
This section conducts a comprehensive performance analysis for passive/active IRS. The closed-form expressions for both scenarios are presented when taking into account the fluctuation of the IRS mounted on the UAV. Since the rate of both UAV fluctuations and channel variations is less than the symbol rate, we will focus on the OP metric that is specified as
\begin{equation}
P_{\text{out}}=\text{Pr}(\gamma_{\text{max}}\leq\gamma_{\text{th}}),
\label{eq.39}
\end{equation}
where $\gamma_{\text{max}}$ and $\gamma_{\text{th}}$ is the maximum instantaneous SNR and the SNR threshold, respectively.

\subsection{Passive-IRS}
The instantaneous maximum SNR of the Passive-IRS system can be calculated by rewriting (\ref{eq.16c}) as
\begin{equation}
\gamma_{\text{max}}=\frac{M\gamma_0\mathbb{G}U}{b_1},
\label{eq.41}
\end{equation}
where $\gamma_0 = \frac{P_t}{\sigma_n^2}$ and $b_1=\frac{P_t\zeta \sigma_e^2}{\sigma_n^2}+1$. Here, we define $U=V^2$ where $V = \sqrt{\beta_0\beta_1(1-\zeta)}\sum_{n=1}^{N}|H_{BI,m,n}||h_{IU,n}|$.
From (\ref{eq.39}) and (\ref{eq.41}), the OP can be obtained as
\begin{subequations}
\begin{align}
P_{\text{out}}&=\text{Pr}(\mathbb{G}U\leq\frac{b_1\gamma_{th}}{M\gamma_0})\\
&=\int\text{Pr}(\mathbb{G}U\leq\frac{b_1\gamma_{th}}{M\gamma_0}|\mathbb{G})f_{\mathbb{G}}(\mathbb{G})d\mathbb{G}\\
&=\int F_U(\frac{b_1\gamma_{th}}{M\gamma_0\mathbb{G}})f_{\mathbb{G}}(\mathbb{G})d\mathbb{G},
\end{align}
\label{eq.42}
\end{subequations}
where $F_U(.)$ is the CDF of random variable $U$. The PDF of $\mathbb{G}$, $f_{\mathbb{G}}(\mathbb{G})$ is derived in Theorem \ref{theorem1} (Eq. \ref{eq.38}). By substituting (\ref{eq.38}) into (\ref{eq.42}), the OP can be given as
\begin{equation}
P_{\text{out}}=\sum_{i=1}^{J}p(i)F_U(\frac{b_1\gamma_{th}}{M\gamma_0q(i)}),
\label{eq.43}
\end{equation}
To determine the CDF of $U$, $F_U(.)$ in (\ref{eq.43}), we estimate $V$, which is the mixture channels associated with IRS, by using the CLT and Gamma approximation.
\begin{theorem}
Under CLT approximation, the closed-form OP for the Passive-IRS system is given by
\begin{equation}
P_{\text{out}}=\sum_{i=1}^{J}p(i)T(q(i)),
\label{eq.44c}
\end{equation}
where
\begin{equation}
T(x)=Q(\frac{\mu_v-\sqrt{\frac{b_1\gamma_{th}}{M\gamma_0x}}}{\sigma_v})-Q(\frac{\mu_v+\sqrt{\frac{b_1\gamma_{th}}{M\gamma_0x}}}{\sigma_v}),
\end{equation}
\begin{equation}
\mu_v = \frac{N\pi c_0\sqrt{1-\zeta}}{4\sqrt{(K_0+1)(K_1+1)}d_0^{\alpha_0/2}d_1^{\alpha_1/2}}L_{\frac{1}{2}}(-K_0)L_{\frac{1}{2}}(-K_1),
\label{eq.45}
\end{equation}
\begin{equation}
\begin{aligned}
\sigma^2_v = \frac{c_0^2N(1-\zeta)}{d_0^{\alpha_0}d_1^{\alpha_1}}(1&-\frac{\pi^2}{16(K_0+1)(K_1+1)}\\
&\quad\quad\quad\times(L_{\frac{1}{2}(-K_0)})^2\times(L_{\frac{1}{2}(-K_1)})^2).
\end{aligned}
\label{eq.46}
\end{equation}
\end{theorem}
\begin{proof}
Given that $U=V^2$, the CDF of $U$ is denoted as
\begin{equation}
\begin{aligned}
F_U(u) &= \text{Pr}(U\leq u)\\
&=F_V(\sqrt{u})-F_V(-\sqrt{u}),
\end{aligned}
\label{eq.47}
\end{equation}
Based on CLT, when $N$ is sufficiently large, $V\sim\mathcal{N}(\mu_v,\sigma_v^2)$. Hence, the CDF of $V$ is calculated as
\begin{equation}
F_V(v) = Q(\frac{\mu_v-v}{\sigma_v}),
\label{eq.48}
\end{equation}
where $\mu_v$ and $\sigma_v$ are the mean and standard deviation of $V$. From (\ref{eq.47}) and (\ref{eq.48}), the CDF of $U$ is represented as:
\begin{equation}
F_U(u) = Q(\frac{\mu_v-\sqrt{u}}{\sigma_v})-Q(\frac{\mu_v+\sqrt{u}}{\sigma_v}).
\label{eq.49}
\end{equation}
Finally, by substituting (\ref{eq.49}) into (\ref{eq.43}), the OP can be obtained.
\end{proof}

Considering the Gamma distribution for $V$, we can derive the closed-form OP for the Passive-IRS given in Theorem 3.
\begin{theorem}
In the case of Gamma approximation, the OP of the Passive-IRS system can be provided as follows
\begin{equation}
P_{\text{out}}=\sum_{i=1}^{J}p(i)(\frac{\gamma(\Lambda,\sqrt{\frac{b_1\gamma_{th}}{M\gamma_0q(i)}}/\Omega)}{\Gamma(\Lambda)}),
\label{eq.57a}
\end{equation}
where $\Lambda=\frac{\mu_v^2}{\sigma_v^2}$ and $\Omega=\frac{\sigma_v^2}{\mu_v}$. $\mu_v$ and $\sigma_v^2$ are provided in (\ref{eq.45}) and (\ref{eq.46}), respectively. Here, $\gamma(.)$ represents the lower incomplete gamma function and $\Gamma(.)$ denotes the complete gamma function.
\end{theorem}
\begin{proof}
The CDF of $U$ is obtained as \cite{li2023irs}
\begin{equation}
F_U(u) = \frac{\gamma(\Lambda,\sqrt{u}/\Omega)}{\Gamma(\Lambda)},
\label{eq.51}
\end{equation}
where $\Lambda=\frac{\mathbb{E}^2(V)}{\text{var}(V)}$ and $\Omega=\frac{\text{var}(V)}{\mathbb{E}(V)}$.
By substituting (\ref{eq.51}) into (\ref{eq.43}), the OP is derived.
\end{proof}

\begin{remark}
For a large number of passive IRS elements ($N \rightarrow \infty$), the scaling $\mathbb{G} \propto 1/N^2$ is counteracted by $U \propto N^2$. Consequently, these scaling laws cancel, rendering the maximum SNR, $\gamma_{\text{max}} = \frac{M\gamma_0 \mathbb{G} U}{b_1}$, independent of $N$ under perfect CSI. As a result, the maximum SNR $\gamma_{\text{max}}$ saturates at a fixed value, independent of $N$. The OP thus converges to a constant floor dictated by UAV fluctuations, resulting in a diversity order of zero. This is in contrast to the case without fluctuations, which achieves an infinite diversity order under perfect CSI. Under imperfect CSI, the OP for both scenarios (with and without fluctuations) converges to a constant floor dictated by the CSI error factor.
\label{remark3}
\end{remark}

\subsection{Active-IRS}
The instantaneous maximum SNR of the Active-IRS system can be obtained by rewriting (\ref{eq.12c}) as
\begin{equation}
\gamma_{\text{max}}=\frac{M\gamma_0\mathbb{G}U}{c_1Z_0+c_2Z_1+c_3},
\label{eq.63}
\end{equation}
where $\gamma_0$ and $U$ defined in (\ref{eq.41}). Here, $c_1=\frac{\sigma_f^2}{\sigma_n^2}\cos^3(\theta_r)$, $c_2=\frac{P_t}{P_F}(1+\frac{P_t\zeta \sigma_e^2}{\sigma_n^2})\cos^3(\theta_t)$, $c_3=\frac{N\sigma_f^2}{P_F}(1+\frac{P_t\zeta \sigma_e^2}{\sigma_n^2})$, $Z_0 = \beta_1\sum_{n=1}^{N}|h_{IU,n}|^2$, and $Z_1=\beta_0\sum_{n=1}^{N}|H_{BI,m,n}|^2$.  Under CLT approximation, $Z_k,~k\in\{0,1\}$ follows Gaussian distribution, specifically $Z_k \sim \mathcal{N}(\mu_{z_k},\sigma^2_{z_k})$ with mean and variance as follows
\begin{subequations}
\begin{align}
\mu_{z_k}&=\frac{Nc_0}{d_k^{\alpha_k}},\\
\sigma^2_{z_k}&=\frac{Nc_0}{d_k^{2\alpha_k}}(\frac{(K_k^2+4K_k+2)}{(K_k+1)^2}-1),~\forall k.
\end{align}
\end{subequations}

\begin{theorem}
In the case of CLT approximation, the OP of the Active-IRS system can be obtained as
\begin{equation}
P_{\text{out}} = \sum_{i=1}^{J}p(i)R(q(i)),
\label{eq.65}
\end{equation}
where
\begin{equation}
R(x)= Q(\frac{\mu_v-\sqrt{\mu_w(x)}}{\sigma_{v}\sqrt{1-\rho^2}})-Q(\frac{\mu_{v}+\sqrt{\mu_w(x)}}{\sigma_{v}\sqrt{1-\rho^2}}),
\end{equation}
\begin{equation}
\mu_w(x)=\frac{\gamma_{\text{th}}}{M\gamma_0x}(c_1\mu_{z_0} + c_2\mu_{z1}+c_3),
\end{equation}
\begin{equation}
\rho = \frac{\mu_{vz}-\mu_v(c_1\mu_{z_0} + c_2\mu_{z1})}{\sigma_v\sqrt{c_1^2\sigma^2_{z_0} + c_2^2\sigma^2_{z1}}},
\end{equation}
\begin{equation}
\mu_{vz} = \sqrt{\beta_0\beta_1(1-\zeta)}(c_1\beta_1A_1 + c_2\beta_0A_0),
\end{equation}
\begin{equation}
\begin{aligned}
&A_p=\frac{3\pi N}{8(1+K_p)^{\frac{3}{2}}(1+K_{\bar{p}})^{\frac{1}{2}}}L_{\frac{3}{2}}(-K_p)L_{\frac{1}{2}}(-K_{\bar{p}})\\
&+\frac{\pi N(N-1)}{4(1+K_0)^{\frac{1}{2}}(1+K_{1})^{\frac{1}{2}}}L_{\frac{1}{2}}(-K_0)L_{\frac{1}{2}}(-K_{1}),\quad p\in\{0,1\}.
\end{aligned}
\end{equation}
\end{theorem}
\begin{proof}
Let $Z=c_1Z_0+c_2Z_1$. Since two RVs $Z_0$ and $Z_1$ are independent, the RV of $Z$ follows Gaussian distribution, i.e., $Z\sim \mathcal{N}(\mu_z,\sigma^2_z)$, which $\mu_Z = c_1\mu_{z_0}+c_2\mu_{z_1}$ and $\sigma^2_z = c_1^2\sigma^2_{z_0}+c^2_2\sigma^2_{z_1}$. Now, we can rewrite SNR in (\ref{eq.63}) as
\begin{equation}
\gamma_{\text{max}} = \frac{M\gamma_0\mathbb{G}U}{Z+c_3}.
\end{equation}
Therefore, the OP can be given as
\begin{equation}
P_{\text{out}}= \int F_{\gamma_{\text{max}}|\mathbb{G}}(\gamma_{\text{th}})f_{\mathbb{G}}(\mathbb{G})d\mathbb{G},
\label{eq.67}
\end{equation}
where $f_{\mathbb{G}}(\mathbb{G})$ is given by (\ref{eq.38}). $F_{\gamma_{max}|\mathbb{G}}(.)$ is the CDF of $\gamma_{\text{max}}$ conditioned on the radiation gain, i.e., $\mathbb{G}$ which is obtained as follows
\begin{subequations}
\begin{align}
F_{\gamma_{max}|\mathbb{G}}(\gamma_{\text{th}})&=\text{Pr}(\frac{U}{Z+c_3}\leq\frac{\gamma_{\text{th}}}{M\gamma_0\mathbb{G}})\\
&=\int\text{Pr}(\frac{U}{z+c_3}\leq\frac{\gamma_{\text{th}}}{M\gamma_0\mathbb{G}}|Z=z)f_Z(z)dz\\
&=\int F_{U|Z}(\frac{\gamma_{th}(z+c_3)}{M\gamma_0\mathbb{G}})f_Z(z)dz\\
&=\mathbb{E}_Z\{F_{U|Z}(w(\mathbb{G},z))\},
\end{align}
\label{eq.68}
\end{subequations}
where $W(\mathbb{G},z)=\frac{\gamma_{th}(z+c_3)}{M\gamma_0\mathbb{G}}$. The distribution of $W$ follows $\mathcal{N}(\mu_w(\mathbb{G}),\sigma_w^2(\mathbb{G}))$ where $\mu_w(\mathbb{G})=\frac{\gamma_{\text{th}}}{M\gamma_0\mathbb{G}}(c_1\mu_{z_0}+c_2\mu_{z_1}+c_3)$ and $\sigma^2_w(\mathbb{G}) = \frac{\gamma^2_{\text{th}}}{M^2\gamma^2_0\mathbb{G}^2}(c_1^2\sigma^2_{z_0}+c_2^2\sigma^2_{z_1})$. From (\ref{eq.49}), we have
\begin{equation}
\begin{aligned}
F_{U|Z}(w(\mathbb{G},z&)) =\\ &Q(\frac{\mu_{v|z}-\sqrt{w(\mathbb{G},z)}}{\sigma_{v|z}})-Q(\frac{\mu_{v|z}+\sqrt{w(\mathbb{G},z)}}{\sigma_{v|z}}).
\label{eq.67b}
\end{aligned}
\end{equation}
$V$ and $Z$ are jointly Gaussian with correlation coefficient $\rho$, the conditional mean $\mu_{v|z}$ and variance $\sigma^2_{v|z}$ are given by
\begin{subequations}
\begin{align}
\mu_{v|z} &= \mu_v + \rho\frac{\sigma_v}{\sigma_z}(z-\mu_z),\\
\sigma^2_{v|z} &= \sigma^2_v(1-\rho^2),
\end{align}
\end{subequations}
where $\rho = \frac{\text{Cov}(v,z)}{\sigma_v\sigma_z}$. Here, $\text{Cov}(v,z) = E\{vz\}-\mu_v\mu_z$ where $E\{vz\}$ is obtained as
\begin{equation}
E\{vz\} = \sqrt{\beta_0\beta_1(1-\zeta)}(c_1\beta_1A_1 + c_2\beta_0A_0).
\end{equation}
The terms $A_p, p\in \{0,1\}$ are calculated as
\begin{equation}
\begin{aligned}
A_p=&\frac{3\pi N}{8(1+K_p)^{\frac{3}{2}}(1+K_{\bar{p}})^{\frac{1}{2}}}L_{\frac{3}{2}}(-K_p)L_{\frac{1}{2}}(-K_{\bar{p}})\\
&+\frac{\pi N(N-1)}{4(1+K_0)^{\frac{1}{2}}(1+K_{1})^{\frac{1}{2}}}L_{\frac{1}{2}}(-K_0)L_{\frac{1}{2}}(-K_{1}).
\end{aligned}
\end{equation}
where $\bar{p}$ denotes the complement of $p$.
Using (\ref{eq.68b}), we approximate $F_{U|Z}(w(\mathbb{G},z))$ in (\ref{eq.67b}) at point $\mu_z$ as follows
\begin{equation}
\begin{aligned}
F_{U|Z}(w(\mathbb{G},z))&\approx F_{U|Z}(w(G,\mu_z))\\
&+ F_{U|Z}^{(1)}(w(G,\mu_z))(z-\mu_z)\\
&+ F_{U|Z}^{(2)}(w(G,\mu_z)) (z-\mu_z))^2.
\end{aligned}
\label{eq.69b}
\end{equation}
By substituting (\ref{eq.69b}) into (\ref{eq.68}), we have
\begin{equation}
\begin{aligned}
F_{\gamma_{max}|\mathbb{G}}(\gamma_{\text{th}})&= F_{U|Z}(\mu_w(\mathbb{G})) \\
&+ F_{U|Z}^{(1)}(\mu_w(\mathbb{G}))\mathbb{E}_Z\{(z-\mu_z)\} \\
&+ F_{U|Z}^{(2)}(\mu_w(\mathbb{G})) \mathbb{E}_Z\{(z-\mu_z)^2\}.
\end{aligned}
\label{eq.70b}
\end{equation}
Given that $\mathbb{E}_Z\{z-\mu_z\}=0$ and the third term in (\ref{eq.70b}) is negligible compared to the first for a sufficient number of IRS elements $N$, we can approximate (\ref{eq.70b}) as
\begin{equation}
F_{\gamma_{max}|\mathbb{G}}(\gamma_{\text{th}})= Q(\frac{\mu_{v}-\sqrt{\mu_w(\mathbb{G})}}{\sigma_{v|z}})-Q(\frac{\mu_{v}+\sqrt{\mu_w(\mathbb{G})}}{\sigma_{v|z}}).
\label{eq.71b}
\end{equation}
Finally, substituting (\ref{eq.71b}) into (\ref{eq.67}) yields the OP.
\end{proof}

\begin{remark}
For a large number of active IRS elements ($N \rightarrow \infty$), the scaling $\mathbb{G} \propto 1/N^2$ is compensated by $U \propto N^2$ in the numerator of $\gamma_{\text{max}}=\frac{M\gamma_0\mathbb{G}U}{Z+c_3}$. Since the terms $Z$ and $c_3$ in the denominator scale proportionally with $N$, $\gamma_{\text{max}} \propto 1/N$. Consequently, the maximum SNR decays to zero under perfect CSI, forcing the OP to converge to $1$. This behavior yields a diversity order of zero, in contrast to the infinite diversity order achievable without fluctuations under perfect CSI. When CSI is imperfect, the OP for both scenarios converges to a constant floor dictated by the CSI error factor.
\label{remark4}
\end{remark}

\begin{remark}
For high SNR ($\gamma_0 \rightarrow \infty$), the diversity order for both Passive-IRS and Active-IRS systems is limited and dictated by the UAV fluctuation level. This is in contrast to the no-fluctuations scenario, which achieves a diversity order of $N$ under perfect CSI \cite{yang2020outage}. When CSI is imperfect, the OP for both scenarios (with and without fluctuations) converges to a constant floor dictated by the CSI error factor, resulting in zero diversity order.
\label{remark5}
\end{remark}

\begin{remark}
The OP for both Passive-IRS and Active-IRS configurations yields the expression $P_{out} = \mathbf{p}^T T(\mathbf{q})$, where $T(\mathbf{q})$ denotes the element-wise application of function $T(\cdot)$ to vector $\mathbf{q}$. Computational evaluation requires: (i) generating vectors $\mathbf{p}$ and $\mathbf{q}$ through Kronecker products of component vectors of length $lD$, (ii) applying the transformation $T(\cdot)$ to each element of $\mathbf{q}$, and (iii) computing the dot product. Given that $J = l^2D^2$ is the dimension of the resulting vectors, each stage has complexity $\mathcal{O}(J)$, establishing the overall computational complexity as $\mathcal{O}(J)$.
\label{remark6}
\end{remark}
\section{Numerical results and discussions}
This section provides numerical results to validate the theoretical analysis of the UMI system. The simulation parameters are configured as follows: the BS, IRS, and UE are located at $[0,0,20]$, $[10,10,120]$, and $[40,40,0]$, respectively. The number of BS antennas is $M = 16$. The path-loss exponents for BS-IRS and IRS-UE links are configured as $\alpha_1 = 2$ and $\alpha_2 = 2.2$, respectively, with a reference path-loss of $c_0 = -30$ dB applied to all links \cite{cheng2023irs}. Additionally, both Rician factors for the BS-IRS and IRS-UE channels are set to $K_0, K_1 = 10$ dB. The thermal noise power of the user's receiver and active IRS is $\sigma_n^2=-80$ dBm and $\sigma_f^2=-70$ dBm, respectively \cite{you2021wireless}. The active IRS amplification power $P_F = 0.05P_t$ and SNR threshold $\gamma_{th} = 10$ dB. The mean and variance of UAV's angular variations along the x-axis and y-axis are the same, with $\mu_x=\mu_y=0^\circ$ and $\sigma_x=\sigma_y=1^\circ$.
To improve comprehension, we plot the CDF instead of the PDF for the UMI radiation pattern. This transition from PDF to CDF provides a better understanding of the radiation patterns and their characteristics.

Fig.~\ref{fig4} shows the CDF of radiation pattern of UMI for both Monte-Carlo simulations as well as analytical results calculated as $F_{\mathbb{G}}(\mathbb{G})=\sum_{i=1}^{J}p(i)U(\mathbb{G}-q(i))$ where $U(x)=\begin{cases} 1, & \mbox{} x\geq0 \\ 0, & \mbox{} x<0 \end{cases}$. Furthermore, Fig.~\ref{fig4} illustrates that the accuracy of the analytical results is highly dependent on the number of sectors $D$ and the number of lobes $l$ (including both main- and side-lobes). A precise alignment between simulation and theoretical results can be achieved with a sufficiently large number of sectors. For instance, the results for $D=60$ ($1 \%$ error) is markedly more accurate than for $D=15$ ($5 \%$ error), at the cost of a $16$-fold increase in complexity, as noted in Remark \ref{remark6}. Furthermore, the analysis for $N=1024$ reveals that the narrow main lobe makes the main-lobe model ($l=1$) insufficient. Consequently, considering only the main-lobe ($l=1$) is insufficient. Accuracy is only achieved when the first side-lobe is included ($l=2$). In contrast, for other array sizes with wider main lobes, considering only the main lobe ($l=1$) provides an accurate approximation.
\begin{figure}[!t]
\centering
\includegraphics[width=3.5in]{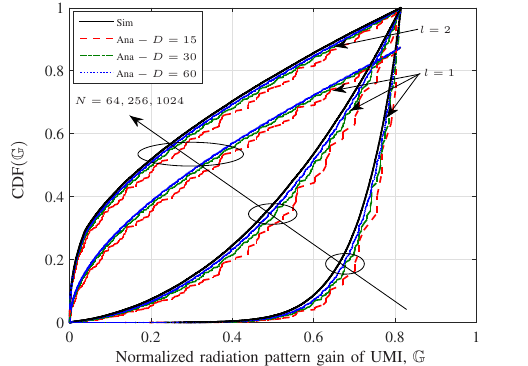}
\caption{The CDF of the radiation pattern of UMI.}
\label{fig4}
\end{figure}

The performance of the UMI system is assessed using the OP as a key performance metric. To analyze the influence of the IRS pattern and the UAV's angular vibrations on system performance, we compare scenarios without vibrations (marked as 'w/o vib' in legends) against those with vibrations (marked as 'vib' in legends). Our analysis reveals a direct trade-off between performance and computational cost governed by the number of sectors $D$. A larger $D$ yields a finer-grained approximation, which consistently improves the approximation accuracy (as measured by the computed OP), at the expense of increased computational complexity, which scales as $O(J)$ where $J = l^2D^2$ as derived in Remark \ref{remark6}.

Fig.~\ref{fig5} shows the OP of the Passive-IRS under CLT and Gamma assumptions, plotted as a function of $P_t$ for varying values of the number of IRS elements $N$ and CSI error factor $\zeta$. The simulation results confirm the accuracy of the derived analytical expressions. As shown in Fig.~\ref{fig5}, the OP decreases with increasing the number of IRS elements $N$ when the UAV experiences no fluctuations. Conversely, when the UAV encounters some fluctuations, the OP increases because of misalignment, as noted in Remark \ref{remark1}. From Fig.~\ref{fig5}, it is clear that a smaller value of elements ($N = 64$) results in improved performance in the high $P_t$ values, while a larger value of elements $N=256$ improves performance at low $P_t$ values. This behavior occurs because, at low $P_t$ values, the weaker IRS directivity (with $N=64$) cannot effectively compensate for the path loss. On the other hand, in high $P_t$ values, the OP for a narrow beam is constrained by the orientation fluctuations of the UAV, as noted in Remark \ref{remark5}. In such cases, a wider IRS pattern is necessary to mitigate the effects of UAV orientation changes. Moreover, the analytical outcome under the CLT and Gamma approximation closely aligns because they are equivalent for large $N$. The number of sectors $D$ in modeling the IRS radiation pattern impact on final results. Specifically, a larger $D=60$ ($2\%$ error) yields a more accurate model compared to a smaller $D=15$ ($8\%$ error). Furthermore, system performance degrades under imperfect CSI ($\zeta=0.1$) compared to the perfect CSI scenario ($\zeta=0$). This effect is more pronounced for a large number of elements, for instance $N = 256$, because the variance of the CSI error is proportional to the cascaded channel gain, which itself scales with $N^2$.
\begin{figure}[!t]
\centering
\includegraphics[width=3.5in]{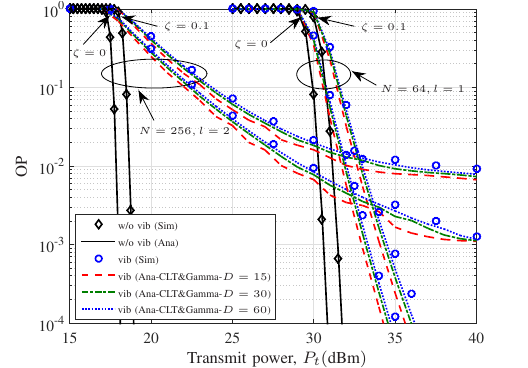}
\caption{The OP of Passive-IRS systems versus transmit power, $P_t$ under CLT and Gamma approximation.}
\label{fig5}
\end{figure}

Fig.~\ref{fig6} illustrates the OP of the Active-IRS system as a function of transmit power, $P_t$, using CLT approximation for varying values of the number of IRS elements $N$ and CSI error factor $\zeta$. Based on our simulation results, we can confirm the accuracy of the derived analytical expression. Consequently, the system utilizing active IRS employs fewer elements due to the amplification capability. As shown in Fig.~\ref{fig6}, the system performance can be maintained relatively stable with a low number of elements, even in the presence of vibrations. This stability can be attributed to the wider IRS pattern associated with fewer elements, which helps mitigate the impact of UAV fluctuations. Additionally, for a given number of IRS elements, the Active-IRS system requires less transmit power than a Passive IRS to achieve the same performance. Alternatively, for a fixed transmit power, the Active-IRS system can achieve comparable performance with fewer elements. Since fewer elements have a wider pattern, this configuration is more robust to UAV fluctuations. From Fig.~\ref{fig6}, we can conclude that when UAVs experience fluctuations, employing active components with fewer elements is beneficial for minimizing the impact of these fluctuations. However, this approach may require additional power consumption. The parameter $D$, which defines the number of sectors in modeling the IRS radiation pattern gain $\mathbb{G}$, impacts the final results. Specifically, a larger $D=60$ ($3\%$ error) yields a more accurate model compared to a smaller $D=15$ ($10 \%$ error). Moreover, the system performance of Active-IRS, like Passive IRS, degrades under imperfect CSI ($\zeta=0.1$) compared to the perfect CSI scenario ($\zeta=0$). This effect is more pronounced for a large number of elements (e.g., $N = 256$) because the CSI error variance scales with $N$ for active IRS but with $N^2$ for passive IRS, underscoring the superiority of the active IRS configuration under imperfect CSI.
\begin{figure}[!t]
\centering
\includegraphics[width=3.5in]{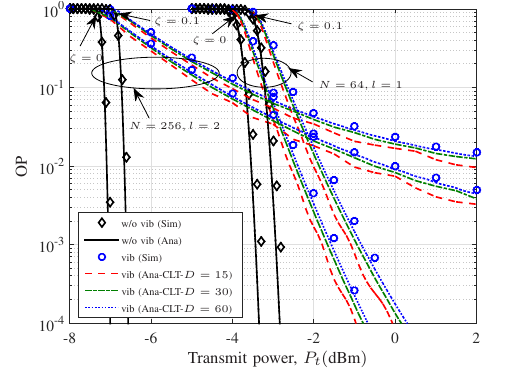}
\caption{The OP of the Active-IRS system versus transmit power, $P_t$, under CLT approximation.}
\label{fig6}
\end{figure}

Since weather conditions change continuously during the day, the angular fluctuations induced by UAVs also change over time.
We propose designing the IRS with the maximum number of elements, denoted as $N_{\text{max}}$, to ensure reliable transmission for the UMI system. As angular fluctuations vary dynamically throughout the day, only a subset of $N$ elements from $N_{\text{max}}$ is utilized. This adaptive approach aims to minimize the OP. Consequently, the optimization problem can be formulated as follows
\begin{equation}
\begin{aligned}
\min_{N} \quad & P_{\text{out}}(N,\mu_x,\mu_y,\sigma_x,\sigma_y)\\
\textrm{s.t.} \quad & 1\leq N\leq N_{\text{max}}.
\end{aligned}
\label{eq.84c}
\end{equation}
The optimal number of reflecting elements along with corresponding minimum achievable outage probabilities can be determined for varying values of $\mu_x$, $\mu_y$, $\sigma_x$, and $\sigma_y$. By analyzing the performance metrics based on these parameters, we can identify the configurations that yield the best reliability for the communication link in the UMI system. It is important to note that, based on our derived analytical expressions, the above optimization problem can be solved efficiently without requiring extensive computational time.

Fig.~\ref{fig7} shows the OP for both passive and active IRS configurations under perfect ($\zeta=0$) and imperfect ($\zeta=0.1$) CSI as a function of the number of elements, $N$. In this analysis, the transmit power is set to $30$ dBm for the passive scenario and $0$ dBm for the active scenario, with $N_{\text{max}}=400$. As illustrated in Fig.~\ref{fig7}, the optimal number of IRS elements for each scenario is determined using our derived analytical expressions. The results show optimal values of $N_{\text{opt}} = 144$ for the Passive-IRS and $N_{\text{opt}} = 49$ for the Active-IRS system. While the Passive-IRS with its optimal configuration achieves the OP of $8\times 10^{-3}$ under imperfect CSI, the Active-IRS achieves a remarkably lower OP of $1\times 10^{-4}$. This significant difference highlights the advantages of using active components in the IRS, particularly in enhancing communication reliability under dynamic conditions. Furthermore, Fig.~\ref{fig7} reveals that the value of $N_{\text{opt}}$ is identical for $D=15$ and $D=30$, emphasizing that a finer sectorization beyond $D=15$ is not necessary for this specific optimization. Additionally, while imperfect CSI ($\zeta=0.1$) leads to a higher overall OP, it does not affect the optimal number of elements; it only shifts the OP performance curve. As shown in Fig.~\ref{fig7}, the diversity order for both IRS systems is zero in the presence of UAV fluctuations, whereas an infinite diversity order is achieved in the absence of fluctuations, even with perfect CSI, as analytically noted in Remarks \ref{remark3} and \ref{remark4}.
\begin{figure}[!t]
\centering
\includegraphics[width=3.5in]{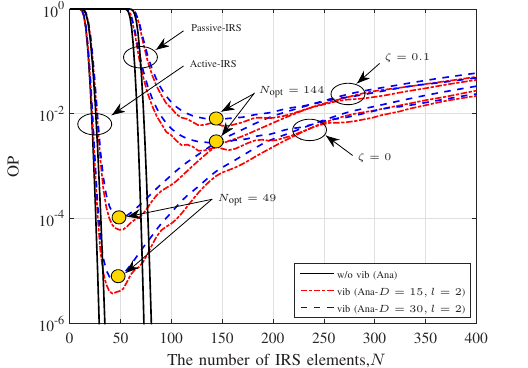}
\caption{The OP of all scenarios versus the number of IRS elements, $N$: finding the optimal number of IRS elements by solving the optimization problem in (\ref{eq.84c}).}
\label{fig7}
\end{figure}

\section{Conclusion}
This study investigated the performance of an IRS mounted on a UAV, characterizing the IRS 3D pattern while accounting for random angular fluctuations of the UAV. We derived a closed-form PDF for the 3D IRS pattern to facilitate performance analysis and validated the accuracy of our analytical models using Monte Carlo simulations. The results demonstrated that, unlike terrestrial IRS systems, aerial IRS performance cannot be enhanced indefinitely by increasing the number of elements due to UAV instabilities. Additionally, simulations revealed that optimal system performance is achieved by selecting the optimal number of elements and incorporating active components. Building upon the foundational understanding established in this work, our future research will extend this framework to a multi-user scenario.


\begin{appendices}
\section{Proof of Lemma 1}
\label{appendix.a}
From (\ref{eq.1}) and (\ref{eq.2}), we rewrite $\varepsilon_{\theta_q}$ as
\begin{equation}
\begin{aligned}
\varepsilon_{\theta_q}(\varepsilon_x,\varepsilon_y) &= \tan^{-1}(\sqrt{\tan^2(\theta_{qx}+\varepsilon_x)+\tan^2(\theta_{qy}+\varepsilon_y)})\\
&-\tan^{-1}(\sqrt{\tan^2(\theta_{qx})+\tan^2(\theta_{qy})}).
\end{aligned}
\label{eq.13}
\end{equation}
For a function of two variables $f(x,y)$ with first-order partial derivatives at a point $(a,b)$, the  first-order Taylor polynomial near $(a,b)$ is given by
\begin{equation}
f(x,y)\approx f(a,b) + f_x(a,b)(x-a) + f_y(a,b)(y-b).
\label{eq.13a}
\end{equation}
Based on (\ref{eq.13a}), we approximate $\varepsilon_{\theta_q}(\varepsilon_x,\varepsilon_t)$ near $(0,0)$ as:
\begin{equation}
\varepsilon_{\theta_q}(\varepsilon_x,\varepsilon_y) \approx A_{\theta_{qx}}\varepsilon_x + A_{\theta_{qy}}\varepsilon_y,
\label{eq.14}
\end{equation}
where $A_{\theta_{qx}}$ and $A_{\theta_{qy}}$ can be expressed as
\begin{equation}
\begin{aligned}
A_{\theta_{qk}}=\frac{(1+\tan^2\theta_{qk})\tan\theta_{qk}}{\sqrt{\tan^2\theta_{qx}+\tan^2\theta_{qy}}(1+\tan^2\theta_{qx}+\tan^2\theta_{qy})},&\\
k\in\{x,y\}.&
\end{aligned}
\end{equation}
From (\ref{eq.14}), since $\varepsilon_{\theta_{q}}$ is the sum of two independent Gaussian RVs, i.e. $\varepsilon_{x}$ and $\varepsilon_{y}$, the distribution of $\varepsilon_{\theta_{q}}$ follows $\varepsilon_{\theta_{q}}\sim \mathcal{N}(\mu_{\varepsilon_{\theta_q}},\sigma^2_{\varepsilon_{\theta_q}})$ where $\mu_{\varepsilon_{\theta_q}}$ and $\sigma^2_{\varepsilon_{\theta_q}}$ are obtained as
\begin{subequations}
\begin{align}
&\mu_{\varepsilon_{\theta_q}} = A_{\theta_{qx}}\mu_x+A_{\theta_{qy}}\mu_y,\\
&\sigma^2_{\varepsilon_{\theta_q}}=A^2_{\theta_{qx}}\sigma^2_x+A^2_{\theta_{qy}}\sigma^2_y.
\end{align}
\end{subequations}

\section{Proof of Lemma 3}
\label{appendix.b}

The RVs $Z_x$ and $Z_y$, which are dependent on UMI fluctuations, follow a Gaussian distribution. From (\ref{eq.28}), the value of sector $i$, defined by $\frac{2i}{DN_q}<|Z_x|\leq\frac{2(i+1)}{DN_q}$, is $\frac{D^2(1-cos(\frac{2\pi i}{D}))}{2\pi^2i^2}$. The probability of $Z_x$ within this sector, $\text{Pr}(\frac{2i}{DN_q}<|Z_x|\leq\frac{2(i+1)}{DN_q})$ is calculated using the CDF of Gaussian-distributed $Z_x$, which can be expressed as $Q$-function. Consequently, the PDF of $g_{ax}$ will be given by
\begin{equation}
f_{g_{ax}}=\sum_{i=0}^{lD-1}{p_{ax}(i)\times\delta(g_{ax}-\frac{D^2(1-cos(\frac{2\pi i}{D}))}{2\pi^2i^2})},
\label{eq.85b}
\end{equation}
where
\begin{equation}
\begin{aligned}
p_{ax}(i)=Q(\frac{2i-DN_q\mu_{Z_x}}{DN_q\sigma_{Z_x}})&-Q(\frac{2(i+1)+DN_q\mu_{Z_x}}{DN_q\sigma_{Z_x}})\\
\qquad+Q(\frac{2i+DN_q\mu_{Z_x}}{DN_q\sigma_{Z_x}})&-Q(\frac{2(i+1)-DN_q\mu_{Z_x}}{DN_q\sigma_{Z_x}}).
\end{aligned}
\label{eq.86b}
\end{equation}
Note that the PDF of RV $g_{ay}$ can be derived similarly to (\ref{eq.85b}) by replacing the subscript $x$ with $y$. According to (\ref{eq.21}), $\mathbb{G}_a=g_{ax}\times g_{ay}$ and applying (\ref{eq.85b}), the PDF of array factor, $\mathbb{G}_a$ conditioned on the array factor in x-axis, $g_{ax}$ can be given as
\begin{equation}
\begin{aligned}
&f_{\mathbb{G}_a|g_{ax}}(\mathbb{G}_a) =\\ &\qquad\quad\sum_{i=0}^{lD-1}\frac{p_{ay}(i)}{g_{ax}}\times\delta(\frac{\mathbb{G}_a}{g_{ax}}-\frac{D^2(1-cos(\frac{2\pi i}{D}))}{2\pi^2i^2}),
\end{aligned}
\label{eq.87b}
\end{equation}
By using (\ref{eq.85b}) and (\ref{eq.87b}), the PDF of array factor, $f_{\mathbb{G}_a}(\mathbb{G}_a)$ is derived as
\begin{equation}
\begin{aligned}
f_{\mathbb{G}_a}(\mathbb{G}_a)&=\int f_{\mathbb{G}_a|g_{ax}}(\mathbb{G}_a) f_{g_{ax}} dg_{ax}\\
&=\sum_{i=0}^{lD-1}\sum_{j=0}^{lD_a-1}{\int \frac{p_{ax}(i)p_{ay}(j)}{g_{ax}}}\\
&\quad\times\delta(\frac{\mathbb{G}_a}{g_{ax}}-\frac{D^2(1-\cos(\frac{2\pi i}{D}))}{2\pi^2i^2})\\
&\quad\times{\delta(g_{ax}-\frac{D^2(1-\cos(\frac{2\pi j}{D}))}{2\pi^2j^2})}dg_{ax}\\
&=\sum_{i=0}^{lD-1}\sum_{j=0}^{lD-1}{p_{ax}(i)p_{ay}(j)\delta(\mathbb{G}_a-q_{ax}(i)q_{ay}(j)))},
\end{aligned}
\label{eq.88b}
\end{equation}
where
\begin{equation}
q_{ak}(i)=\frac{D^2(1-\cos(\frac{2\pi i}{D}))}{2\pi^2i^2},\quad k\in \{x,y\}.
\label{eq.89b}
\end{equation}
We can rewrite result in (\ref{eq.88b}) as
\begin{equation}
f_{\mathbb{G}_a}(\mathbb{G}_a)=\sum_{i=1}^{l^2D^2}p_a(i)\delta(\mathbb{G}_a-q_a(i)),
\end{equation}
where $\mathbf{q}_a = \mathbf{q}_{ax}\otimes \mathbf{q}_{ay}$ and $\mathbf{p}_a = \mathbf{p}_{ax}\otimes \mathbf{p}_{ay}$. Here, $\mathbf{q}_{ak}\in \mathbb{R}^{lD\times1}$ and $\mathbf{p}_{ak}\in \mathbb{R}^{lD\times1}$ where $k\in\{x,y\}$. $p_{ak}(i)$ and $q_{ak}(i)$ are given in (\ref{eq.86b}) and (\ref{eq.89b}), respectively.

\end{appendices}


\bibliographystyle{IEEEtran}
\bibliography{mybib}


\vfill

\end{document}